%% file: FSAI.tex
\documentclass{article}

\input{header.tex}
\usepackage{times}
\usepackage{hyperref}
\usepackage{url}
\usepackage{algorithm}
\usepackage{algorithmic}
\usepackage{eqparbox}
\usepackage{graphicx}
\usepackage{float}

\usepackage{pgfplots} 
\pgfplotsset{compat=newest} 
\pgfplotsset{plot coordinates/math parser=false} 
\newlength\figureheight 
\newlength\figurewidth 

\newcommand{\sr}{\,\mathrm{sr}}
\newcommand{\tsr}{\,\tilde{sr}}

\newcommand{\svrg}{\textsc{SVRGsolve}}
\newcommand{\accsvrg}{\textsc{AccSVRGsolve}}

\author{Alon Gonen\footnote{School of Computer Science, The Hebrew University, Jerusalem, Israel}  \and Shai Shalev-Shwartz\footnote{School of Computer Science, The Hebrew University, Jerusalem, Israel}}

\title{Faster Low-rank Approximation using Adaptive Gap-based Preconditioning}

\begin{document}

\maketitle

\begin{abstract} 
We propose a method for rank $k$ approximation to a given input matrix $X \in \reals^{d \times n}$ which runs in time 
\[
\tilde{O} \left(d ~\cdot~ \min\left\{n +\tsr(X)\,G^{-2}_{k,p+1} ~,~ n^{3/4}\,\tsr(X)^{1/4}\,G^{-1/2}_{k,p+1}\right\} ~\cdot~\poly(p)\right) ~,
\] 
where $p>k$, $\tsr(X)$ is related to stable rank of $X$, and $G_{k,p+1} = \frac{\sigma_k-\sigma_p}{\sigma_k}$ is the multiplicative gap between the $k$-th and the $(p+1)$-th singular values of $X$. In particular, this yields a linear time algorithm if the gap is at least $1/\sqrt{n}$ and $k,p,\tsr(X)$ are constants.
\end{abstract}

\section{Introduction} \label{sec:intro} 
We consider the fundamental \emph{low-rank approximation} problem: given $X  \in \reals^{d \times n}$, a target dimension $k<d$ and an accuracy parameter $\epsilon$, we would like to find a rank-$k$ orthogonal projection matrix $\Pi$ which approximately minimizes the error $\|X-\Pi X\|_\xi$, where $\|\cdot\|_\xi$ is either the Frobenius norm $\|\cdot\|_F$ or the spectral norm $\|\cdot\|$. This problem has many important applications in \emph{machine learning}. The most prominent example is \emph{Principal Component Analysis} (PCA); when the columns of $X$, denoted $x_1,\ldots, x_n \in \reals^d$, are data points, the Frobenius norm error coincides with the objective of PCA.

Denote the SVD of $X$ by $X= \sum \sigma_i u_i v_i^\top$, where the singular values are sorted in a descending order. It is well known that the projection matrix minimizing $\|X - \Pi X\|_\xi$, for both Frobenius and spectral norm, is $\Pi^\star = \sum_{i=1}^k u_i u_i^\top$, and we have $\Pi^\star X = X_k:= \sum_{i=1}^k \sigma_i u_i v_i^\top$. Therefore, the best rank $k$ approximation can be found by SVD computation. However, computing the SVD is often prohibitively expensive, and we therefore seek for efficient approximation algorithms.

Naturally, this problem has received much attention in the literature. For the case $k=1$, a well known approach is Power iteration (\cite{trefethen1997numerical}), which starts with some random vector in $\reals^d$, and keep multiplying it by $X^\top X$, while normalizing the vector after each such multiplication. A generalization to $k > 1$ is obtained by starting with a random matrix of size $d \times p$, keep multiplying it by  $X^\top X$, while ortho-normalizing the matrix after each such multiplication. The analysis of the Power iteration depends on the multiplicative spectral gap --- for every $j > i$ we denote 
\begin{equation} \label{eq:condition}
  G_{i,j} = \frac{\sigma_i-\sigma_{j}}{\sigma_i} ~.
\end{equation}
In particular, \cite{musco2015stronger} have shown that the Power iteration finds a (multiplicative) $(1+\epsilon)$-approximate minimizer after at\footnote{We use the notation $\tilde{O}$ to hide polylogarithmic dependencies. In particular, here it hides a $\log(1/\epsilon)$ factor.} most $\tO(G^{-1}_{k,k+1})$ iterations. We view the quantity $G^{-1}_{k,k+1}$ as the \emph{condition number} of the problem. The runtime at each iteration is $O(dn)$, hence the total runtime is $\tO(dn \,G^{-1}_{k,k+1})$.\footnote{We note that if $X$ is sparse, the runtime at each iteration is controlled by the number of nonzero elements in $X$, denoted $\nnz(X)$. Hence, the term $dn$ can be replaced by $\nnz(X)$. In the formal statement of our results we use the more tight bound of $\nnz(X)$, but for the simplicity of the representation, throughout the introduction we stick to the dense case.} If the multiplicative gap value is a constant then the runtime of Power iteration becomes linear in the size of $X$. However, when the gap is small, the runtime becomes too large. Several recent algorithmic ideas and new analyzes of existing algorithms have lead to significant improvements over Power iteration.
\begin{enumerate}
\item \textbf{Oversampling: } The Power iteration starts with an initial random matrix $S \sim \cN(0,1)^{d \times k}$ and keep multiplying it by $X X^\top$. Even though we are interested in the top $k$ eigenvectors, we can apply the Power iteration with a matrix $S \sim \cN(0,1)^{d \times p}$, where we refer to $p > k$ as the oversampling parameter. After approximately finding the top $p$ eigenvectors of $X X^\top$, we project the columns of $X$ onto the subspace spanned by these $p$ eigenvectors, and find the top $k$ eigenvectors in the reduced subspace. The runtime of the latter stage is negligible (as long as $p$ is of the same order as $k$). Several recent papers analyzed the effect of oversampling on the convergence rate (e.g., \cite{halko2011finding,musco2015stronger,wang2015improved}), showing that now the required number of iterations is order of $G^{-1}_{k,p+1} $ rather than order of $G^{-1}_{k,k+1}$. A common empirical observation is that while the gap between the consecutive singular values $\sigma_k$ and $\sigma_{k+1}$ might be tiny, we often can find $p$ of the same order of magnitude as $k$ such that the gap between $\sigma_k$ and $\sigma_{p+1}$ is substantially larger.
\item \textbf{Matrix-free shift-and-invert preconditioning: } The \emph{shift-and-invert} method is a well established \emph{preconditioning} technique in numerical linear algebra (\cite{trefethen1997numerical}). Roughly speaking, for some appropriately chosen \emph{shift parameter} $\lambda$, this preconditioning process reduces the task of approximating several eigenvectors of $A=XX^\top$ to the task of approximating several eigenvectors of $D=(\lambda I-A)^{-1}$. For example, note that if $0<\lambda-\lambda_1$, then the top eigenvector of $D$ coincides with $u_1$, the top eigenvector of $A$. Furthermore, it is seen that if $\lambda-\lambda_1=a (\lambda_1-\lambda_2)$ for some positive constant $a$, then the multiplicative gap between the first and the second eigenvalue of $D$ becomes a constant. Consequently, for such a choice, by applying the Power iteration to $D$ rather than to $A$, we converge to $u_1$ rapidly. The catch is that inverting $(\lambda I-A)$ is as costly as an exact SVD computation. On the other hand, since the Power iteration only requires multiplications with $(\lambda I-A)^{-1}$, it makes sense to avoid the inversion and approximate each such multiplication to an high accuracy. This is exactly the approach taken by \cite{garber2015fast} and \cite{jin2015robust}. In particular, by slightly modifying the Stochastic Reduced Variance Gradient (SVRG) algorithm due to \cite{johnson2013accelerating}, they were able to approximately solve each linear system to an extremely high accuracy in time $\tilde{O} (d(n + \sr(X) G_{{1,2}}^{-2}))$, where 
\begin{equation} \label{eqn:stable_rank_def}
\sr(X) := \|X\|_F^2/\sigma_1^2 
\end{equation}
is the stable rank of $X$. Since the Power iteration applied to $D$ requires only polylogarithmic number of iterations in order to converge to $u_1$, the overall complexity is dominated by the complexity of a single application of SVRG.  

Comparing the obtained runtime to the Power iteration, we observe that this method has a worse dependence on the gap, $G_{1,2}^{-2}$ vs. $G_{1,2}^{-1}$, and an additional dependence on the stable rank, $\sr(X)$. However, the advantage is that $\sr(X) G_{1,2}^{-2}$ is being \emph{added} to $n$ rather than \emph{multiplied} by $n$. As a result, this method is much faster than Power iteration whenever $\sr(X) G_{1,2}^{-1} \ll n$.
\item \textbf{Acceleration:} 
The Lanczos method, which has been recently analyzed in \cite{musco2015stronger}, reduces the number of iterations of Power iteration to order of $G^{-1/2}_{k,k+1}$, and yields a runtime of order $d n G^{-1/2}_{k,k+1}$. There is a close relationship between this improvement to Nesterov's accelerated gradient descent. In fact, for the case of $k=1$, by using an acceleration version of SVRG (\cite{frostig2015regularizing}), the complexity of the ``Matrix-free shift-and-invert preconditioning'' method described previously becomes $\tilde{O} (d \, n^{3/4}\, (\sr(X))^{1/4} \, G^{-1/2}_{1,2} )$. 
\end{enumerate}

The goal of this paper is to develop a method that enjoys all of the above three improvements and that is not restricted to the case $k=1$. 

The first step is to inject oversampling into the ``Matrix-free shift-and-invert preconditioning'' method, so that its runtime will depend on $G_{1,p+1}$ rather than on $G_{1,2}$. As will be apparent soon, this is obtained by using Power iteration (see \secref{sec:twoStage}) instead of the vanilla Power method, while using SVRG to approximately compute $p$ matrix-vector products rather than $1$ on each round. 
While this step is technically easy, it is important from practical perspective, as in many cases, the gap between the first and second eigenvalues is small, but there is a constant $p$ such that the gap between the first and the $(p+1)$-th eigenvalues is large.

The second step is to generalize the results for $k > 1$. A naive approach is to rely on a \emph{deflation} technique, namely, to approximate one eigenvector at a time using the algorithm for the case $k=1$ while removing components that we have already computed.  As mentioned by \cite{shamir2015fast}, the problem with this approach is that both the convergence rate and the success of the deflation procedure heavily depend on the existence of large enough gaps between \emph{all} of the top leading eigenvalues, which will usually lead to a long runtime. 
Instead, we suggest an adaptive algorithm which estimates the gaps between the leading singular values and based in this information, it divides the low-rank approximation task into several smaller subproblems. Depending on the condition number of each subproblem, our algorithm chooses between direct application of the Power iteration and an extension of the ``Matrix-free shift-and-invert preconditioning''. 

To summarize, we strengthen the results of \cite{garber2015fast} and \cite{jin2015robust} in two important ways: a) while their results are limited to the task of approximating the top leading eigenvector, our results apply to any target dimension. b) we allow the incorporation of oversampling techniques that lead to further improvements in terms of gap dependence. This makes the method more practical and suitable to large-scale eigenvalue problems.

\subsection{Our Results} \label{sec:results}
The next theorem formally states our contribution. We denote by the number of non-zero elements of $X$ by $\nnz(X)$. The definition of $G_{i,j}$ is given in \eqref{eq:condition} and of $\sr(X)$ is given in \eqref{eqn:stable_rank_def}.
\begin{theorem} \label{thm:mainResult}
Let $X \in \reals^{n \times d}$ and let $1 \le k < p < d$ be such that $\sigma_k-\sigma_{p+1}>0$. Denote by $\tsr(X)=\max_{i \in \{1,\ldots,k-1\}}\sr(X-X_i)$. For any $\delta, \epsilon \in (0,1)$, with high probability, our algorithm finds an orthogonal rank-$k$ projection matrix $\hPi$ which satisfies 
\[
\|X-\hPi X\|_\xi \le (1+\epsilon) \|X-X_k\|_\xi~,
\]
(where $\|\cdot\|_\xi$ is either the Frobenius or the spectral norm) in time $\tilde{O} ((\nnz(X) + d \,\tsr(X) G_{{k,p+1}}^2)\, \poly(p))$ or $\tilde{O} \left(\left(\nnz(X)^{3/4} (d \tsr(X))^{1/4} \sqrt{G_{{k,p+1}} } \right)\, \poly(p) \right)$ if acceleration is used.
\end{theorem}
Few comments are in order. When $k=p=1$, our bounds are identical to the bounds of \cite{jin2015robust}. The computational price of extending the result to any $k$ and $p$ is polynomial in $p$, as one could expect. As we mentioned above, by using oversampling we may substantially improve the gap dependency. Finally, while in general $\tilde{\sr}(X)$ and $\sr(X)$ are not comparable, they have the same roles in the cases $k=1$ and $k>1$, respectively. Namely, both are upper bounded by the rank of $X$. Furthermore, as we are interested in reducing the dimensionality to $k$, we implicitly presume that $\sr(X)$ and $\tsr(X)$ are much smaller than the rank in the cases $k=1$ and $k>1$, respectively.

\subsection{Related work} \label{sec:related}
The low-rank approximation problem has also been studied in a gap-independent setting. As was shown in recent papers (\cite{musco2015stronger}, \cite{wang2015improved}), although one can not hope to recover the leading eigenvectors in this setting, the Power iteration and the Lanczos methods yield the same norm bounds in time $\tO(\epsilon^{-1} \nnz(X)p)$ and $\tO(\epsilon^{-1/2} \nnz(X)p)$, respectively. 

Recently there has been an emerging interest in randomized methods for low-rank approximation (\cite{woodruff2014sketching, rokhlin2009randomized, shamir2015convergence}) both in offline, stochastic and streaming settings (\cite{jain2016streaming, jin2015robust}). Furthermore, some of these methods share the important advantage of decoupling the dependence on $\nnz(X)$ from the other dependencies. In the gap independent setting, the sketch-and-solve approach (\cite{sarlos2006improved, woodruff2014sketching, clarkson2013low}) yields the fastest methods which run in time $\nnz(X)+\poly(\sr(X),\epsilon)$. Unfortunately, no linearly convergent algorithms can be obtained using this approach in the gap-dependent setting. 

Another approach is to use randomization in order to perform cheaper updates relative to Power iteration. The simplest algorithm which uses one random column of $X$ at a time is called Oja's algorithm (\cite{oja1985stochastic}). The basic idea is that for a random column $x_j$, the rank-$1$ matrix $x_j x_j^\top$ forms an unbiased estimate of $A=XX^\top$. Due to the noise arising from the estimation process, Oja's method is not a linearly convergent algorithm. Recently, \cite{shamir2015convergence} used variance reduction techniques to remedy this situation. In some sense, the method proposed by \cite{shamir2015convergence} is to Oja's method as Stochastic Variance Reduced Gradient (SVRG) is to Stochastic Gradient Descent (SGD). It should be remarked that the low-rank minimization problem is substantially non-convex. Nevertheless, \cite{shamir2015convergence} was able to obtain a linearly convergent algorithm while decoupling the size of the data from the condition number. While this method is suitable to any $k>1$, as explained in detail in \cite{jin2015robust}, the bounds of \cite{shamir2015convergence} have suboptimal dependence on the natural parameters. Furthermore, no accelerated bounds are known for this algorithm. Last, while the reduction approach taken here and in \cite{garber2015fast, jin2015robust} allows us to easily incorporate any further improvements to Power iteration (e.g., the oversampling idea), it is not clear how to integrate these results into the analysis of \cite{shamir2015convergence}.

\paragraph{Organization}
In \secref{sec:pre} we introduce the notation used throughout the paper and discuss some preliminaries. \secref{sec:main} is devoted to the description of our algorithm. Missing proofs can be found in the Appendix.

\section{Preliminaries} \label{sec:pre}
\subsection{Notation} \label{sec:notation}
We denote by $t_C$ the time it takes to multiply a matrix $C$ by a vector. For $p \le d$, we denote by $\cO^{d \times p}$ the set of $d \times p$ matrices with orthonormal columns. Given $X \in \reals^{n \times d}$ whose SVD is $X=U\Sigma V^\top = \sum_{i=1}^{\min\{n,d\}} \sigma_i u_i v_i^\top$, the best rank-$k$ approximation to $X$ (w.r.t. both $\|\cdot\|_2$ and $\|\cdot\|_F$) is $X_k = U_k \Sigma_k V_k = \sum_{i=1}^k \sigma_i u_i v_i^\top$. We denote the reminder $X-X_k$ by $X_{-k}$. Let $k < p < d$ and suppose that $Y \in \reals^{d \times p}$ has a full column rank. We often need to compute the best rank-$k$ approximation to $X$ in the column space of $Y$. For the Frobenius norm error, a closed-form solution is given by $Q(Q^\top M)_k$, where $Q \in \reals^{d \times p}$ is an orthonormal matrix whose span coincides with the range of $Y$ (\cite{woodruff2014sketching}[Lemma 4.1]).

\subsection{Power iteration: A two-stage framework for low-rank approximation}  \label{sec:twoStage}
In this section we describe a basic two-stage framework for $k$-rank approximation (\cite{halko2011finding, hardt2014noisy, balcan2016improved}) which we simply call Power iteration. Recall that we aim at finding an approximated low-rank approximation to the matrix $X = U \Sigma V^\top \in \reals^{d \times p}$. The matrix $X$ can be thought of as the data matrix presented at the beginning, or alternatively, a deflated data matrix resulted from a removal of the top components (which have already been approximately computed). 
\subsubsection{First stage}
The input in the first stage consists of a semidefinite matrix $C \in \reals^{d \times d}$ whose eigenvectors are equal to the left singular vectors of $X$, and an oversampling parameter $p$. While the natural choice is $C = XX^\top$, we sometimes prefer to work with a different matrix mainly due to conditioning issues. The method iteratively multiplies a randomly drawn matrix $S \in \cN(0,1)^{d \times p}$ from left by $C$ and ortho-normalizes the result (see \algref{alg:subIter}).\footnote{Usually, $C$ has some factored form (e.g., $C=XX^\top$). In such a case we do not form the matrix $C$ when performing multiplications with $C$.} The runtime of each iteration is $t_C \cdot p +dp^2$, where the latter term is the cost of the QR factorization.
\begin{algorithm}
\caption{First stage of power iteration: subspace iteration}
\label{alg:subIter}
\begin{algorithmic}
\STATE \textbf{Input: } A positive semidefinite matrix $C \in \reals^{d \times d},p,L$~~ $(1 < p < d)$
\STATE Draw $S^{(0)} \in \cN(0,1)^{d \times p}$
\FOR {$\ell=1$ \TO $L$}
\STATE $Y^{(\ell)} = C S^{(\ell-1)}$
\STATE $Y^{(\ell)} = S^{(\ell)}  R^{(\ell)}$ (QR decomposition)
\ENDFOR
\STATE \textbf{Output: $S^{(L)}$}
\end{algorithmic}
\end{algorithm}
An elegant notion that captures the progress during this stage is the principal angles between subspaces (we provide the definition and some basic properties in the Appendix).
\begin{theorem}(\textbf{\cite{wang2015improved}[Theorem 6.1]}) \label{thm:zhang} 
Let $C = U \barLam U^\top \succeq 0$, $k$ be a target dimension and $p>k$ be the oversampling parameter. Suppose that we run \algref{alg:subIter} with the input $(C,p)$. Then with high probability, after $L= O(G_{k,p+1}^{-1}\log(d/\epsilon)$ iterations, we have
\[
\tan(\theta_k(U_k,S^{(L)}))\le \epsilon~.
\]
\end{theorem}

\subsubsection{Second stage}
The first stage yields a matrix $S^{(L)} \in \cO^{d \times p}$ whose range is approximately aligned with the leading eigenvectors of $C$, as reflected by \thmref{thm:zhang}. In the second stage we use $S^{(L)}$ to compute the Frobenius best rank-$k$ approximation to $X = U \Sigma V^\top$ in the column space of $S^{(L)}$ (see \secref{sec:notation}). The complexity is $O(pdn)$.
\begin{algorithm}
\caption{Second stage of Power iteration: low-rank approximation restricted to a subspace}
\label{alg:angleToNormAlg}
\begin{algorithmic}
\STATE \textbf{Input: } A positive semidefinite matrix $C \in \reals^{d \times d},S \in \cO^{d \times p},k~~(k<p<d)$
\STATE Compute the eigenvalue decomposition $S^\top C S = \hat{U} \hat{\Sigma} \hat{U}^\top \in \reals^{p \times n}$
\STATE \textbf{Output: } Return the matrices $S, \hat{U}_k$ which form the projection matrix $S \hat{U}_k \hat{U}_k^\top S^\top$
\end{algorithmic}
\end{algorithm}
There are standard techniques for translating principal angle bounds into matrix norm bounds (e.g. \cite{wang2015improved}[Section 6.2]). We will employ such a technique in the proof of our main theorem.

\subsection{SVRG/SDCA based solvers for linear systems} \label{sec:SVRG}
As mentioned above, \cite{jin2015robust} suggested a variant of SVRG for minimizing convex quadratic sum of non-convex functions. They also derived an accelerated variant.  We call the corresponding methods \svrg~and \accsvrg. We next state the complexity bounds for these methods. 

\begin{theorem} \label{thm:svrgSolve} \textbf{(\cite{jin2015robust}[Theorems 12,15])}
Let $X \in \reals^{d \times n}$ and $\lambda$ be a shift parameter such that $0<\lambda-\lambda_1$, where $\lambda_1=\lambda_1(XX^\top)$. Denote by $D = (\lambda I-XX^\top)$. For any vector $b$ and $\epsilon,\delta \in (0,1)$, with probability at least $1-\delta$, \svrg$(X,\lambda,b)$ returns a vector $x$ with $\|x-(\lambda I - XX^\top)^{-1}b\|_D \le \epsilon$ in time $\tilde{O} \left (\nnz(X)+d \, \sr(X) \frac{\lambda_1^2}{(\lambda-\lambda_1)^2} \right )$. 
The \accsvrg~method satisfies the same accuracy conditions in time $\tilde{O} \left (\nnz(X)^{3/4} (d \,\sr(X))^{1/4}  \frac{\lambda_1^{1/2}}{(\lambda-\lambda_1)^{1/2}} \right)$.
\end{theorem}
As we explain in \secref{sec:precision}, throughout this paper we implicitly use these results whenever we consider matrix-vector products with shifted-inverse matrices. 
\subsection{Gap-independent  approximation of eigenvalues} \label{sec:musco}
We use the following gap-independent bounds due to \cite{musco2015stronger} for estimation of eigenvalues using Power iteration. 
\begin{algorithm}
\caption{Gap-independent eigenvalues approximation}
\label{alg:musco}
\begin{algorithmic}
\STATE \textbf{Input: } $C \succeq 0, k, \epsilon~~(k<d)$
\STATE Run \algref{alg:subIter} with the input $(C,k,L=O(\epsilon^{-1} \log(d/\epsilon)))$ to obtain $\tU$
\STATE Run \algref{alg:angleToNormAlg} with the input $(C,\tU,k)$ to obtain $Z=S\hU_k$
\STATE For all $i \in [k]$ compute $\hlam_i = z_i^\top Cz_i$
\STATE \textbf{Output: } $\hlam_1,\ldots,\hlam_k$
\end{algorithmic}
\end{algorithm}
\begin{theorem} \label{thm:musco} \textbf{(\cite{musco2015stronger}[Theorem 1]}
Let $C \in \reals^{d \times d}$ be a positive semidefinite matrix, $k<d$ and $\epsilon \in (0,1)$ be the input to \algref{alg:musco}. Then, with high probability, the output of the algorithm satisfies
\[
(1-\epsilon) \lambda_i(C) \le \hlam_i \le (1-\epsilon)^{-1} \lambda_i(C)
\]
for all $i \in \{1,\ldots,k\}$. The runtime is  $\tilde{O} (\epsilon^{-1} \, t_C)$.
\end{theorem}

\subsection{Precision and high probability bounds} \label{sec:precision}
In order to simplify the presentation we make the following two assumptions: a) The deflation procedure is accurate, i.e., whenever we approximately compute the eigenvectors $u_1,\ldots,u_{s-1}$ and proceed to handle the remaining $k-s+1$ components, the projection of $X$'s columns onto the orthogonal complement to $\{u_1,\ldots,u_{s-1}\}$ is accurate. b) Whenever we use \svrg~or~\accsvrg~to approximately compute matrix-vector products with shifted-inverse matrices, the returned solution is accurate. Since both our method for approximating the eigenvectors and \svrg~are linearly convergent methods, these two assumptions hold in any reasonable computing environment.\footnote{Our assumption is analogous to the usual assumption that exact methods such as the QR algorithm (\cite{trefethen1997numerical}) can find the SVD of $X$ in time $O(nd^2)$ (this assumption is used in the analysis of both Power iteration and Lanczos). As mentioned in \cite{musco2015stronger}, the Abel-Ruffini Theorem implies that an exact SVD is incomputable. Nevertheless, such assumptions are reasonable once we establish high accuracy methods that converge rapidly to the exact solution.} Furthermore, the (theoretical) challenge of taking into account the noise arising from both procedures can be carried out using the established framework of noisy Power iteration (\cite{hardt2014noisy, balcan2016improved}) while incurring only polylogarithmic computational overhead.\footnote{This is essentially the approach taken in \cite{jin2015robust}.} 

There is only one source of randomization in our algorithm, namely the initialization of \algref{alg:subIter}. Since we use this algorithm $\tO(\poly(p))$ times and since the probability of failure scales like $\exp(-d)$, our statements hold with high probability.

\section{Gap-based Approach for Low-Rank Approximation} \label{sec:main}
In this section we describe our algorithm in detail and prove the main result. We assume that we are given as an input a parameter $\Delta>0$ which satisfies 
$$
\Delta \le \lambda_k-\lambda_{p+1} \le 2\Delta~.
$$
Note that we can find such a $\Delta$ with negligible incurred runtime overhead of $O\left(\log \left(\frac{1}{\lambda_k-\lambda_p} \right) \right)$. We view the parameter $\Delta$ as a ``Gap Budget''. Indeed, as will become apparent soon, one can adjust $\Delta$ and the oversampling parameter $p$ in accordance.

\subsection{The Partitioning strategy}
Assume that we already computed the first $s-1$ leading eigenvectors of $A=XX^\top$, $u_1,\ldots, u_{s-1}$. Denote by 
$$
I_0 = \{1,\ldots,k\},~I_{\textrm{prev}} = \{1,\ldots,s-1\}~,I= I_0 \setminus I_{\textrm{prev}} = \{s,\ldots,k\}~.
$$ 
Assume that the deflation is accurate, i.e., we already applied the projection $(I-\sum_{i=1}^{s-1}u_i u_i^\top)$ to the columns of the input matrix $X$. We would like to extract a subinterval of the form $\{s,\ldots,q\} \subseteq I$ such that the gap between $\lambda_q$ and the proceeding eigenvalues would allow us to compute the eigenspace corresponding $\{s,\ldots,q\}$ reasonably fast. We distinguish between several gap scales:
\begin{enumerate}
\item
We first seek for a (multiplicative) gap\footnote{We interchangeably refer both to the (multiplicative) gaps between the $\sigma_i$'s and to the gaps between the $\lambda_i$'s. It is easily seen that the corresponding expressions are of the same order of magnitude.} of order $\poly(1/p)$. If we find such a gap then we use the Power iteration (without neither preconditioning or oversampling) to approximate $u_s,\ldots,u_q$ in time $\tilde{O}(nd \,\poly(p))$.
\item
Otherwise, we seek for an additive gap of order $\Delta$. If we find such a gap, then we use the shifted inverse Power iteration (without oversampling) to extract $u_s,\ldots,u_q$. As we shall see, by requiring that $q$ is the minimal index in $I$ with this property and choosing a shift $\lambda$ with $\lambda - \lambda_s = a\Delta$ for some constant $a \in (0,1)$, we ensure that the multiplicative gap between the corresponding eigenvalues of $\lambda I-A$ is $O(\poly(p))$. Also, the fact that we have not found a multiplicative gap of order $1/p$ implies that $\lambda_S$ and $\lambda_k$ are of the same order, hence the runtime of SVRG scales with the ``right'' gap (see \corref{cor:mainCond}).
\item
Otherwise, we simply return $q=k$. We will then use the shifted-inverse Power iteration with oversampling in order to utilize the gap of order $\Delta$ between $\lambda_k$ and $\lambda_{p+1}$.
\end{enumerate}
Obviously, one difficulty is that we do not know the eigenvalues. Hence, we will derive estimates both of the multiplicative and the additive gaps.

\subsubsection{Searching for multiplicative gaps of order $\poly(1/p)$} \label{sec:largeDetect}
By applying \algref{alg:musco} to the deflated matrix $A_{-(s-1)} = (I-\sum_{i=1}^{s-1} u_i u_i^\top)A(I-\sum_{i =1}^{s-1} u_i u_i^\top)$ with target dimension $k-s+1$ and accuracy $\epsilon' = 1/(9p^2)$, we obtain $\hlam_s,\ldots,\hlam_k$ which satisfy
\[
(1-\epsilon') \lambda_i \le  \hlam_i \le (1-\epsilon')^{-1} \lambda_i~~~\textrm{for all}~ i \in I~.
\]
(note that we refer to the indices of the matrix $A$ before deflation). Based on these estimates, we can detect gaps of order $\poly(1/p)$.
\begin{lemma} \label{lem:largeDetect}
Suppose that $\hlam_{i+1} \le \hlam_i (1-p^{-2})$. Then, $\lambda_{i+1} \le \lambda_i (1-p^{-4})$. Conversely, if $\lambda_{i+1} \le \lambda_i (1-p^{-1})$, then $\hlam_{i+1} \le \hlam_i (1-p^{-2})$. 
\end{lemma}
\lemref{lem:largeDetect} suggests the following simple partitioning rule: if exists, return any $q$ with $\hlam_{q+1} \le \hlam_q (1-p^{-2})$ (see \algref{alg:largeDetect}).
\begin{algorithm}
\caption{Detection of multiplicative gaps of order $\poly(1/p)$}
\label{alg:largeDetect}
\begin{algorithmic}
\STATE \textbf{Input: } $I=\{s,\ldots,k\}, \hlam_s, \ldots, \hlam_q$
\STATE If exists, return any $q \in I \setminus \{k\}$ which satisfies $\hlam_{q+1} \le \hlam_q (1-p^{-2})$
\STATE Otherwise, return $-1$
\end{algorithmic}
\end{algorithm}
We  deduce the following implication.
\begin{corollary} \label{cor:largeDetect}
Suppose that the partitioning procedure returns $q \in I$ with $\hlam_{q+1} \le \hlam_q (1-p^{-2})$. Then the condition number when applying the Power iteration to $A_{-(s-1)}$ with target dimension $k-s+1$ (and no oversampling) is $\poly(p)$. Conversely, if the procedure does not find such $q$, then $\lambda_k \ge \lambda_s/10$.
\end{corollary}
\subsubsection{Searching for additive gaps} \label{sec:deltaDetect}
In the absence of a multiplicative gap of order $\poly(1/p)$, we turn to search for additive gaps of order $\Delta$. Since we prefer to avoid applying the Power iteration with an accuracy parameter of order $\Delta$, we need to employ a more sophisticated estimating strategy. To this end, we  develop an iterative scheme that updates a shift parameter $\lambda$ in order to obtain better approximations to the gaps between the eigenvalues. Let $\lambda \in [\lambda_s+\Delta, 2\lambda_s]$ be\footnote{Recall that we assume that $\lambda_s \ge \lambda_k \gg \Delta$ (otherwise conditioning is not needed).} an initial shift parameter. Such a $\lambda$ can be easily found by applying \algref{alg:musco} to $A_{-(s-1)}$ (see \secref{sec:tuneApp} in the appendix). Consider the deflated shifted matrix 
\begin{equation} \label{eq:deflatedShift}
D_{-(s-1)} = (I - \sum_{i=1}^{s-1} u_i u_i^\top) (\lambda I - A)^{-1} (I- \sum_{i=1}^{s-1} u_i u_i^\top) ~.
\end{equation}
By applying \algref{alg:musco} to $D_{-(s-1)}^{-1}$ a with a target dimension $k-s+1$ and a reasonably large accuracy parameter $\epsilon' = \frac{1}{9p}$, we find $\tlam_s,\ldots, \tlam_k$ which satisfy
$$
(1-\epsilon') (\lambda - \lambda_i)^{-1} \le \tlam_i \le (1-\epsilon')^{-1} (\lambda-\lambda_i)^{-1}~~\textrm{for all}~~i \in I~.
$$
By inverting, we obtain the following approximation to $\lambda-\lambda_i$:
\begin{equation} \label{eq:approxShiftGap}
(1-\epsilon') \tlam_i^{-1} \le  \lambda-\lambda_i \le (1-\epsilon')^{-1} \tlam_i^{-1}~~\textrm{for all}~i \in I~.
\end{equation}
Since for any $q \in I \setminus \{k\}$, $\lambda_q-\lambda_{q+1} = (\lambda-\lambda_{q+1}) - (\lambda-\lambda_q) \approx \tlam_{q+1}^{-1} - \tlam_q^{-1}$, we can derive upper and lower bounds on the gaps between consecutive eigenvalues. Based on these bounds, in \algref{alg:deltaDetect} we suggest a simple partitioning rule.
\begin{algorithm}
\caption{Detection of additive gaps of order $\Delta$}
\label{alg:deltaDetect}
\begin{algorithmic}
\STATE \textbf{Input: } $I=\{s, \ldots,k\}$, $\Delta$, $\tlam^{-1}_i$ for all $i \in I$
\IF {$J=\{q' \in I \setminus \{k\}: \tlam^{-1}_{i+1} - \tlam^{-1}_i \ge \frac{5}{9} \Delta\} \neq \emptyset$}
\STATE Return $q = \min J$ 
\ELSE
\STATE Return $q=k$ 
\ENDIF
\end{algorithmic}
\end{algorithm}
The success of this method depends on the distance between $\lambda$ and $\lambda_s$. Specifically, our analysis requires that
\begin{equation} \label{eq:shiftClose} 
\frac{\Delta}{27}  \le \tlam^{-1}_s \le \frac{\Delta}{5}    \underbrace{\Rightarrow}_{\epsilon' <1/10} \frac{\Delta}{30}  = \frac{9}{10} \cdot \frac{\Delta}{27} \le \lambda-\lambda_s \le \frac{10}{9} \cdot \frac{\Delta}{5}  = \frac{2 \Delta}{9}~.
 \end{equation} 
Inspired by \cite{garber2015fast,jin2015robust}, in \secref{sec:tuneApp} we describe a an efficient method which enforces (\ref{eq:shiftClose}) by iteratively deriving constant approximations to $\lambda-\lambda_s$ and updating the shift accordingly. Assuming that (\ref{eq:shiftClose}) holds, we turn to prove the correctness of the partitioning rule. The next lemma implies that gaps of desired magnitude are identified by our method.
\begin{lemma} \label{lem:deltaDetect}
Let $\beta > 0$. Suppose that for some $q \in I \setminus \{k\}$, $\lambda_q-\lambda_{q+1} \ge \beta$ and $q$ is the minimal index with this property. Then, $\tlam^{-1}_{q+1} - \tlam^{-1}_q \ge \frac{5}{9} \beta$. 
\end{lemma}
The following lemma shows that gaps detected by our method are indeed sufficiently large.
\begin{lemma} \label{lem:deltaSound}
Suppose that our method returns $q$ with $q<k$. Then, $\lambda_q-\lambda_{q+1} \ge \Delta/9$.
\end{lemma}
As mentioned above, in the case that $q<k$, we will be interested in the gap between the $q$-th and the $(q+1)$-th eigenvalues of $D^{-1}$. Otherwise, we will be interested in the gap between the $k$-th and the $(p+1)$-th eigenvalues. Thus, we define
\[
\tG  = 
\begin{cases}  
(\lambda_q(D^{-1}) - \lambda_{q+1} (D^{-1}))/\lambda_q(D^{-1}) & q < k \\ 
(\lambda_k(D^{-1}) - \lambda_{p+1} (D^{-1}))\lambda_p(D^{-1}) & q = k
\end{cases}
\]
\begin{corollary} \label{cor:mainCond}
Assume that $\lambda$ satisfies (\ref{eq:shiftClose}) and let $q$ be the output of \algref{alg:deltaDetect}. The condition number when applying the Power iteration to the shifted inverse matrix $D_{-(s-1)}^{-1}$ (\ref{eq:deflatedShift}) is $G^{-1} = O(k)=O(p)$. Furthermore, the complexities of \svrg~and \accsvrg~when applied to approximately compute matrix-vector multiplications with the matrix $D_{-(s-1)}^{-1}$ are $\tilde{O} ((\nnz(X)+d \,\sr(X_{-(s-1)}) G_{k,p+1}^{-2})k)$ and $\tilde{O} ((\nnz(X)^{3/4}(d \,\sr(X_{-(s-1)}))^{1/4} G_{k,p+1}^{-1/2})k)$, respectively.
\end{corollary}

\noindent \textbf{Tuning the shift parameter: } In \algref{alg:shiftTune} we suggest a simple method that yields a shift parameter $\lambda$ and a corresponding estimate $\tlam^{-1}_s$ that satisfy (\ref{eq:approxShiftGap}). We defer the description and the analysis of this method to the appendix.

\subsection{The Algorithm}
All the pieces are in place. Our algorithm (see \algref{alg:main}) iteratively combines the partitioning procedure with the corresponding application of Power iteration. We turn to prove the main result. We start by stating a slightly weaker result.
\begin{theorem} \label{thm:mainResultW}
Let $X \in \reals^{n \times d}$ be the input matrix and let $k$ and $p$ be the target dimension and the oversampling parameter, respectively $(k < p <d)$. Suppose that $\sigma_k-\sigma_p>0$ and define $G_{k,p+1}$ as in (\ref{eq:condition}). For any $\delta,\epsilon \in (0,1)$, with probability at least $1-\delta$, \algref{alg:main} finds an orthogonal rank-$k$ projection matrix $\hPi$ which satisfies 
\[
\|X-\hPi X\|_F \le \|X-X_k\|_F +\epsilon\|X\|_F~,
\]
in time $\tilde{O} ((\nnz(X) + d \,\tsr(X) G_{{k,p+1}}^2)\, \poly(p))$ or $\tilde{O} ((\nnz(X))^{3/4} (d \tsr(X))^{1/4} \sqrt{G_{{k,p+1}} })\, \poly(p))$ if acceleration is used. 
\end{theorem}
One usually expect to see error bounds that scale with $\epsilon\|X-X_k\|_F$ rather than with $\epsilon \|X\|_F$. Since the dependence of the runtime on $1/\epsilon$ is logarithmic, this is not an issue in our case. From the same reason, it is easy to establish also spectral norm bounds. Indeed, note that at the beginning of \algref{alg:main}, the accuracy parameter $\epsilon$ is rescaled according to some rough upper bound on $\|X\|_F/\|X_{-k}\|$.\footnote{Such an estimate can be easily obtained using \algref{alg:musco}.} The reason for this scaling operation is now apparent.
\begin{corollary} \label{cor:mainResult}
Under the same conditions as in \thmref{thm:mainResult}, suppose that that we know a rough upper bound $\mu$ on $\|X\|_F/\|X-X_k\|_F$. By modifying the given accuracy parameter, we ensure that with probability at least $1-\delta$, 
\[
\|X-\hPi X\|_\xi \le (1+\epsilon) \|X-X_k\|_\xi ~
\]
where $\|\cdot\|_\xi$ is either the Frobenius or the spectral norm. The runtime overhead relative to the complexity bound in \thmref{thm:mainResult} is logarithmic in $\mu d$.
\end{corollary}

\begin{proof} \textbf{(of \thmref{thm:mainResultW})} \\
\textbf{Correctness: }
Each iteration $j \in [t]$ corresponds to an interval of the form $I_j=\{s_j,\ldots,q_j\}$. For each $j \in [t]$, denote by $k_j=|I_j|$ and let $U^{(j)} \in \reals^{d \times k_j}$ the matrix consisting of the columns $s_j,\ldots,q_j$ of $U$. Using \thmref{thm:zhang} along with the bounds on the condition number established in \corref{cor:largeDetect} and \corref{cor:mainCond}, we see that each time we invoke \algref{alg:subIter}, we obtain $\tU^{(j)}$ with
$$
\tan (\theta_{k_j}(U^{(j)},\tU^{(j)})) \le \epsilon/k~,
$$
Let $\tU = [\tU^{(1)},\ldots,\tU^{(t)} ] \in \reals^{d \times p}$. Our strategy is to show the existence of a rank-$k$ approximation to $X$ in the column space of $\tU$ which satisfies the accuracy requirements. Since we return the optimal rank-$k$ approximation to $X$ in the column space of $\tU$, this will imply the desired bound. 

Recall that we denote by $\cP_k$  the set of all $p \times p$ rank-$k$ projection matrices. Note that for $j=1,\ldots, t-1$, $\tU^{(j)}$ and $U^{(j)}$ have the same number of columns, whereas $\tU^{(t)}$ has $p-k$ more columns than $U^{(t)}$. Let
\[
P = \argmin_{P \in \cP_k} \tan (\theta_{k_t}(U^{(t)},\tU^{(t)}P))
\]
Let $P=ZZ^\top$, where $Z$ has orthonormal columns. For $j < t$, denote $\tPi_j = \tU^{(j)}(\tU^{(j)})^\top$ and let $\tPi_t = (\tU ^{(t)} Z)(\tU^{(t)}Z)^\top$. We now consider the rank-$k$ orthogonal projection 
$$
\tPi = \sum_{j=1}^t \tPi^{(j)}~
$$
Using the triangle inequality we obtain that
\begin{align*}
\|X - \tPi X \|_F &\le \|(I-\tPi) X_{-k}\|_F +  \sum_{i=1}^t \|(I- \sum_{j=1}^t \tPi^{(j)} )U^{(i)} \Sigma^{(i)} V^{(i)}\|_F \\
&\le \|X-X_k\|_F +  \sum_{i=1}^t \|(I- \sum_{j=1}^t \tPi^{(j)} ) U^{(i)} \Sigma^{(i)} V^{(i)}\|_F~,
\end{align*}
where the last inequality follows from the fact that for any matrix $Y$ and any projection matrix $\Pi$, $\|Y\|_F \ge \|\Pi Y\|_F$. We turn to bound each of the summands on the right-hand side. We use the following fact: if $\Pi$ and $\Pi'$ are two $d \times d$ projection matrix such that the range of $\Pi'$ contains the range of $\Pi$, then for any $d \times n$ matrix $M$, $\|M-\Pi' M\|_F \le \|M-\Pi M\|_F$. For each $i \in [t]$ this fact implies that
\[
\|(I- \sum_{j=1}^t \tPi^{(j)} ) U^{(i)}\Sigma^{(i)} V^{(i)}\|_F \le \|(I- \tPi^{(i)} )U^{(i)}  \Sigma^{(i)} V^{(i)}\|_F~.
\]
Using the unitary invariance and the submultiplicativity of the Frobenius norm, we further bound this term by
\begin{align*}
\|(I- \tPi^{(i)} )U^{(i)}  \Sigma^{(i)} V^{(i)}\|_F &= \|(I- \tPi^{(i)} )U^{(i)}  \Sigma^{(i)}\|_F \le \|\Sigma^{(i)} \|\|(I- \tPi^{(i)} )U^{(i)} \|_F \\
&\le \|X\|_F \sin (\theta_{k_i}(U^{(i)},\tU^{(i)})) \le \|X\|_F \tan (\theta_{k_i}(U^{(i)},\tU^{(i)})) \le   \|X\|_F (\epsilon/k)~.
\end{align*}
Combining the bounds above we obtain the claimed bound
\[
\|X - \tPi X \|_F  \le \|X-X_k\|_F + \frac{t \epsilon \|X\|_F}{k} \le \|X-X_k\|_F + \epsilon \|X\|_F~.
\]
\noindent \textbf{Runtime: } We analyze the unaccelerated case. The analysis for the accelerated case is analogous. The main algorithm runs for $t$ iterations, each of which corresponds to a single subinterval. Clearly, $t \le k$. For each subinterval we call Power iteration polylogarithmic number of times. According to \corref{cor:largeDetect} and \corref{cor:mainCond}, for each of this calls, the condition number associated with Power iteration is $\tO(\poly(p))$. This implies the same bound on the number of iterations. When applied to matrices of the form $A_{-(s-1)}$ the complexity per iteration is $O(\nnz(X) \poly(p))$. When applied to shifted inverse matrices of the form $D_{-(s-1)}^{-1}$, the complexity is controlled by the complexity of \svrg. By \corref{cor:mainCond}, this complexity scales with $\tilde{O} (\nnz(X)+d \,\sr(X) G_{k,p+1}^{-2})$.
\end{proof}

\begin{proof} \textbf{(of \corref{cor:mainResult})}
By replacing $\epsilon$ with $\epsilon \mu/(3d)$, we obtain 
\[
\|X-\tPi X\|_F \le (1+\epsilon/(3d))\|X-X_k\|_F
\]
This already gives the desired Frobenius bound. Squaring both sides yields
\[
\|X-\tPi X\|_F^2 \le (1+\epsilon/(3d))^2\|X-X_k\|_F^2 \le (1+\epsilon/d)\|X-X_k\|_F^2 \le \|X-X_k\|_F^2 +\epsilon \|X-X_k\|_2^2~,
\]
Since additive (squared) Frobenius norm bound implies spectral additive norm bound (\cite{musco2015stronger}[Lemma 15]), we obtain
\[
\|X-\tPi X\|^2_2 \le \|X-X_k\|_2^2 +\epsilon \|X-X_k\|_2^2 = (1+\epsilon)\|X-X_k\|_2^2
\]
Taking the square root of both sides yields the desired bound.
\end{proof}

\begin{algorithm}[ht!]
\caption{low-rank Approximation using Adaptive Gap-based Preconditioning}
\label{alg:main}
\begin{algorithmic}
\STATE \textbf{Input: } $X \in \reals^{d \times p}, 1 \le k < p < d$, $ \Delta \le \sigma_k^2 - \sigma_p^2 \le 2\Delta, \epsilon$
\STATE $s=1$, $t=0$, $\epsilon = \epsilon/ (\mu kd)$ where $\mu \ge \|X\|/\|X_{-k}\|$
\WHILE {$q \le k$}
\STATE $I=I_{\textrm{rem}} = \{s,\ldots,k\}$, $t=t+1$
\STATE $X_{-(s-1)} = (I-\sum_{i=1}^{s-1} \tu_i \tu_i^\top)X,~A_{-s+1} = X_{-s+1}X_{-s+1}^\top$
\STATE Apply \algref{alg:musco} with input $(A_{-(s-1)},k-s+1,1/(9p^2))$ to obtain $\{\hlam_i: i \in I\}$
\STATE Run \algref{alg:largeDetect} with the input $(I,\{\hlam_i: i \in I\})$ to obtain $q$
\IF {$q \neq -1$}
\STATE Run \algref{alg:subIter} with $(A_{-(s-1)},q,L = O(p^4 \log(kd/\epsilon)))$ to obtain $\tU^{(t)} = [\tu_s,\ldots, \tu_q]$
\ELSE
\STATE Run \algref{alg:shiftTune} (from \secref{sec:tuneApp}) with the input $(I,\Delta)$ to obtain $\lambda$ 
\STATE Define $D_{-s+1}$ as in (\eqref{eq:deflatedShift})
\STATE Apply \algref{alg:musco} with the input $(D_{-(s-1)},k-s+1,1/(9p))$ to obtain $\{\tlam_i: i \in I\}$
\STATE Run \algref{alg:deltaDetect} with the input $(I,\Delta, \{\tlam: i \in I\})$ to obtain $q$
\STATE If $q=k$ set $p'=p$. Otherwise, set $p'=p$
\STATE Run \algref{alg:subIter} with $(D_{-(s-1)},p',L = O(k \log(kd/\epsilon)))$ to obtain $\tU^{(t)}= [\tu_s,\ldots, \tu_q]$
\ENDIF
\STATE $s=q+1$
\ENDWHILE
\STATE Form the $d \times p$ matrix $\tU = [\tU^{(1)}, \ldots , \tU^{(t)}]$
\STATE Run \algref{alg:angleToNormAlg} with the input $(XX^\top,\tU,k)$ to obtain the final projector $\tPi=\tU \hat{U}_k \hat{U}_k^\top \tU^\top$
\STATE \textbf{Output: } $\tPi$
\end{algorithmic}
\end{algorithm}

\section*{Acknowledgments} 
We thank Ohad Shamir and Alon Cohen for useful discussions. The work is supported by ICRI-CI and by the European Research Council (TheoryDL project).

\newpage 

\bibliography{bib}
\bibliographystyle{plain}

\newpage
\appendix

\appendix

\section{Proofs}

\subsection{Proofs from \secref{sec:largeDetect}}
\begin{proof} \textbf{(of \lemref{lem:largeDetect})}
Suppose that that $\hlam_{i+1} \le \hlam_i (1-p^{-2})$. Then
\begin{align*}
\lambda_{i+1} \le (1-\epsilon')^{-1} \hlam_{i+1} \le (1-\epsilon')^{-1} (1-p^{-2}) \hlam_i \le (1-\epsilon')^{-2} (1-p^{-2}) \lambda_i~.
\end{align*}
Since $\epsilon'<1/3$, $(1-\epsilon')^{-2} \le 1+9\epsilon' \le 1+p^{-2}$. It follows that 
\[
\lambda_{i+1} \le (1+p^{-2}) (1-p^{-2}) \lambda_i = (1-p^{-4}) \lambda_i~.
\]
Conversely, suppose that that $\lambda_{i+1} \le \lambda_i (1-p^{-1})$. Then
\begin{align*}
\hlam_{i+1} \le (1-\epsilon')^{-1} \lambda_{i+1} \le (1-\epsilon')^{-1} (1-p^{-1}) \lambda_i \le (1-\epsilon')^2 (1-p^{-1}) \hlam_i~.
\end{align*}
We already know that $(1-\epsilon')^2 \le 1+9 \epsilon' \le 1+p^{-2} < 1+p^{-1}$. Therefore,
\[
\hlam_{i+1} \le (1+p^{-1}) (1-p^{-1}) \hlam_i = (1-p^{-2}) \hlam_i~.
\]
\end{proof}

\begin{proof} \textbf{(of \corref{cor:largeDetect})}
The first part follows immediately from the first part of \lemref{lem:largeDetect}. If such a gap is not found then it follows from the second part of \lemref{lem:largeDetect} that $\lambda_{i+1} \ge \lambda_i(1-p^{-1})$ for all $i \in I \setminus \{k\}$. Hence, using the equality $1-x \ge \exp(-2x)$ which holds for all $x \in (0,1/2)$, we obtain
\[
\lambda_k \ge (1-p^{-1}) \lambda_{k-1} \ge \ldots \ge (1-p^{-1})^{k-s} \lambda_s \ge (1-p^{-1})^{p} \lambda_s \ge e^{-2} \,\lambda_s \ge \lambda_s/10~. 
\]
\end{proof}

\subsection{Proofs from \secref{sec:deltaDetect}}
We start by deriving upper and lower bounds on $\lambda_i - \lambda_{i+1}$ for all $i \in I \setminus \{k\}$. Using (\ref{eq:approxShiftGap}) together with the fact that $\lambda_i - \lambda_{i+1} = (\lambda - \lambda_{i+1}) - (\lambda-\lambda_i)$, we obtain
\begin{align} \label{eq:upperGap}
\lambda_i-\lambda_{i+1} &\le (1-\epsilon')^{-1} \tlam^{-1}_{i+1} - (1-\epsilon') \tlam^{-1}_i \notag \\
&=  (1-\epsilon')^{-1}(\tlam^{-1}_{i+1} - \tlam^{-1}_i) + ((1-\epsilon')^{-1} - (1-\epsilon')) \tlam_i^{-1}~.
\end{align}
Similarly, we obtain the following lower bound:
\begin{align} \label{eq:lowerGap}
\lambda_i-\lambda_{i+1} & \ge (1-\epsilon') \tlam^{-1}_{i+1} - (1-\epsilon')^{-1} \tlam^{-1}_i \notag \\
&=  (1-\epsilon')(\tlam^{-1}_{i+1} - \tlam^{-1}_i) - ((1-\epsilon')^{-1} - (1-\epsilon')) \tlam_i^{-1}~.
\end{align}
We turn to prove the lemmas.
\begin{proof} \textbf{(of \lemref{lem:deltaDetect})}
Denote by $\mu = (1-\epsilon')^{-1} - (1-\epsilon')$ and note that since $\epsilon' \le 1/9$, $\mu \le 3 \epsilon'$. Next we note that by the minimality of $q$, $\lambda_s - \lambda_q \le k \beta$. Hence, using (\ref{eq:approxShiftGap}) we deduce the bound
\begin{equation} \label{eq:estCumGap}
\tlam_q^{-1} \le (1-\epsilon')^{-1} k \beta~.
\end{equation}
A rearrangement of (\ref{eq:upperGap}) yields
\[
\tlam^{-1}_{q+1} - \tlam^{-1}_q \ge (1-\epsilon')\left((\lambda_q-\lambda_{q+1})-\mu \tlam_q^{-1} \right) \ge  (1-\epsilon')(\beta - \mu \tlam_q^{-1})~.
\]
Together with (\ref{eq:estCumGap}), we obtain
\begin{align*}
\tlam_{q+1}^{-1} - \tlam_q^{-1} &\ge  (1-\epsilon')(\beta -  \mu(1-\epsilon')^{-1} k \beta)  = \beta(1-\epsilon'-\mu k) \ge \beta(1-\epsilon'-3 \epsilon' k) \ge \frac{5}{9} \beta~,
\end{align*}
where we substituted $\epsilon'= \frac{1}{9p}< 1/9$.
\end{proof}

\begin{proof} \textbf{(of \lemref{lem:deltaSound})}
We follow the notation from the previous proof. The partitioning rule (\algref{alg:deltaDetect}) implies that 
$$
\tlam^{-1}_{q+1} - \tlam^{-1}_q \ge \frac{5}{9} \Delta~.
$$
According to the lower bound (\ref{eq:lowerGap}), we have
\[
\lambda_q- \lambda_{q+1} \ge (1-\epsilon')(\tlam^{-1}_{q+1} - \tlam^{-1}_q) - \mu \tlam_q^{-1} \ge \frac{8}{9} \cdot \frac{5}{9} \Delta - \mu \tlam_q^{-1}~,
\]
where we substituted $\epsilon' = \frac{1}{9p} \le \frac{1}{9}$. We proceed to bound $\tlam_q^{-1}$ from above. Recalling the assumption $\tlam_s^{-1} \le \frac{\Delta}{5}$ and using the minimality\footnote{Note that by the minimality of $q$, each of the summands below are smaller than $\frac{5}{9}\Delta $.} of $q$, we obtain
\[
\tlam_q^{-1} = \tlam_s^{-1} + (\tlam_{s+1}^{-1} - \tlam_s^{-1}) + \ldots + (\tlam_q^{-1} - \tlam_{q-1}^{-1}) \le \Delta \left(\frac{1}{100} + \frac{5k}{9} \right) \le \frac{2k}{3}  \Delta~.
\]
Combining the bounds yields
\[
\lambda_q- \lambda_{q+1} \ge \frac{8}{9} \cdot \frac{5}{9} \Delta - 3 \epsilon' \cdot \frac{2k}{3}  \Delta \ge \Delta \left( \frac{40}{81} - \frac{2}{9}\right) \ge \Delta/9~,
\]
where we used again the fact that $\mu \le 3\epsilon'$ and substituted $\epsilon'$.
\end{proof}

\begin{proof} \textbf{(of \corref{cor:mainCond})}
We first establish the claimed upper bound on $\tG^{-1}$. Consider first the case where $q<k$. Note that
\begin{align*}
\tG^{-1} &= \frac{\frac{1}{\lambda-\lambda_q}}{\frac{1}{\lambda-\lambda_q} - \frac{1}{\lambda-\lambda_{q+1}}}  = \frac{\lambda-\lambda_{q+1}}{\lambda_q-\lambda_{q+1}}  = \frac{\lambda-\lambda_s}{\lambda_q-\lambda_{q+1}} + \frac{\lambda_s-\lambda_q}{\lambda_q-\lambda_{q+1}} + \underbrace{\frac{\lambda_q-\lambda_{q+1}}{\lambda_q-\lambda_{q+1}}}_{=1}
\end{align*}
According to \lemref{lem:deltaSound}, $\lambda_q-\lambda_{q+1} \ge \Delta/9$. Also, by assumption $\lambda-\lambda_s \le 2\Delta/9$. Finally, using the minimality of $q$ together with \lemref{lem:deltaDetect}, we obtain that $\lambda_s - \lambda_q \le (q-s) \Delta \le k \Delta$. Combining the bounds, we see that
\[
\tG^{-1} = O(k) = O(p)~.
\]
The same bound applies for the case where $q=k$; the bound $\lambda_q-\lambda_{q+1} \ge \Delta/9$ is replaced by the bound $\lambda_k - \lambda_{p+1} \ge \Delta$ (by assumption) and the same upper bounds hold (by exactly the same arguments).

We proceed to bound the complexity of \svrg. Note that the leading eigenvalue in our case is $\lambda_s$. Also, multiplication with $X_{-(s-1)}$ can be done in time $\nnz(X) k$ (first multiply by $X$ and then project). Using \thmref{thm:svrgSolve} we obtain that each matrix multiplication costs
\[
\tilde{O} \left( \left(\nnz(X)+ d \sr(X) \frac{\lambda_s^2}{(\lambda-\lambda_s)^2}  \right ) k \right)~.
\]
By assumption $\lambda-\lambda_s=a \Delta$ for some constant $a$. Next, we recall that we resort to shifted-inverse preconditioning only if we did not find a multiplicative gap of order $1/p$. It follows from \corref{cor:largeDetect} that $\lambda_s \le 10 \lambda_k$. Multiplying the right term in the above bound by the constant $\lambda_k^2/\lambda_s^2$, we see that the complexity of \svrg~is at most
\[
\tilde{O} \left( \left(\nnz(X)+ d \sr(X) \frac{\lambda_k^2}{\Delta^2}  \right ) k \right)
 = \tilde{O} ((\nnz(X)+d \,\sr(X_{-(s-1)}) G_{k,p+1}^{-2})k)~,
\]
where we substituted $\lambda_k/\Delta = \Theta(G_{k,p+1}^{-1})$. The analysis for the accelerated solver is analogous.
\end{proof}

\section{Tuning the shift parameter} \label{sec:tuneApp}
Recall that the initial shift parameter satisfies $\lambda-\lambda_s \in [\Delta, \lambda_s]$. Applying the Power iteration to $D_{-s+1}^{-1}$ with target dimension $1$ and $\epsilon'=1/10$ yields $\tlam^{-1}_s$ which satisfies (\ref{eq:approxShiftGap}). Hence, initially, $\tlam^{-1}_s \in [\frac{9}{10} \Delta, \frac{10}{9} \lambda_s]$, i.e., $\lambda_s^{-1}$ does not lie in the desired range $[\Delta/27,\Delta/5]$. As we formally prove below, by iteratively performing an update of the form $\lambda_+ = \lambda - \frac{(1-\epsilon')}{2} \tlam_s^{-1}$ and re-estimating $\lambda_+-\lambda_s$, we ensure that both $\lambda_+-\lambda_s$ and its corresponding estimate $(\tlam^{-1}_s)_+$ decrease by a constant factor. Furthermore, the constants we chose ensure that $\tlam^{-1}_s$ will eventually fall into the desired range $[\Delta/27,\Delta/5]$. The procedure is detailed in \algref{alg:shiftTune} and its correctness follows from the following lemma.
\begin{lemma} \label{lem:shiftTune}
\algref{alg:shiftTune} terminates after at most $O(\log(\lambda_s/\Delta))$ iterations. Upon termination, $\lambda-\lambda_s$ and the corresponding estimate $\tlam^{-1}_s$ satisfy \eqref{eq:shiftClose}.
\end{lemma}
\begin{algorithm}
\caption{Shift tuning}
\label{alg:shiftTune}
\begin{algorithmic}
\STATE \textbf{Input: } $I = \{s,\ldots,k\},~\Delta$
\STATE Apply \algref{alg:musco} with the input $(A_{-s+1},1, 1/4)$ to obtain $\hlam_s$
\STATE Set $\lambda = (1+1/2) \hlam_s$, 
\STATE $\tlam^{-1}_s = +\infty$
\WHILE {$\tlam^{-1}_s \notin [\Delta/27,\Delta/25]$}
\STATE Define $D_{-s+1}$ as in (\ref{eq:deflatedShift})
\STATE Run \algref{alg:musco} with the input parameters $(D_{-(s-1)}^{-1} , 1,\epsilon'=1/9)$ to obtain $\tlam^{-1}_s$
\STATE $\lambda = \lambda - \frac{(1-\epsilon')}{2} \tlam_s^{-1}$
\ENDWHILE
\STATE Return $\lambda$
\end{algorithmic}
\end{algorithm}

\begin{proof} \textbf{(of \lemref{lem:shiftTune})}
As mentioned above, the assumption on the initial shift and the approximation guarantees imply that the initial estimate $\tlam^{-1}_s \ge 9\Delta/10 > \Delta/5$, i.e., $\tlam^{-1} \notin [\Delta/27,\Delta/5]$. Denote $\lambda_+ = \lambda-\frac{(1+\epsilon)^{-1}}{2} \tlam^{-1}_s$. The following inequalities indicate that we preserve the positivity of the gap while decreasing it by a multiplicative constant factor:
\[
\lambda_+ - \lambda_s  = \lambda-\lambda_s-\frac{(1-\epsilon')}{2}  \tlam_s^{-1} \ge \lambda-\lambda_s-\frac{1}{2} (\lambda-\lambda_s) = \frac{\lambda-\lambda_s}{2}~,
\]
\[
\lambda_+ - \lambda_s = \lambda-\lambda_s-\frac{(1-\epsilon')}{2}  \tlam_s^{-1} \le (\lambda-\lambda_s) (1-(1-\epsilon')^2/2) \le  (\lambda-\lambda_s)(1-81/200) \le \frac{3}{4} (\lambda-\lambda_s)~,
\]
where we substituted $\epsilon'=1/9$. 

It is left to verify that $\tlam^{-1}_s$ also decreases by a constant factor at each iteration and that it eventually falls into the desired range (the bound on the number of iterations will follow from the first two claims). Indeed, in the next step the algorithm updates the deflated inverse matrix $D_{-s+1}$ with the new shift parameter $\lambda_+$ and invokes \algref{alg:musco} to $D_{-s+1}^{-1}$ to obtain a new estimate $(\tlam^{-1}_s)_+$ which satisfies
\begin{align*}
(\tlam^{-1}_s)_+ &\in \left [\frac{9}{10} (\lambda_+-\lambda_s),\frac{10}{9}(\lambda_+-\lambda_s)  \right] \subseteq \left[\frac{9}{10} \cdot \frac{\lambda-\lambda_s}{2} , \frac{10}{9} \cdot \frac{3(\lambda-\lambda_s)}{4}  \right] \\
&  \subseteq \left[\frac{81}{100} \cdot \frac{\tlam^{-1}_s}{2} , \frac{100}{81} \cdot \frac{3\tlam^{-1}_s}{4}  \right]  \subseteq  \left[\frac{\tlam^{-1}_s}{4}, \frac{25 \tlam_s^{-1}}{27} \right]~.
\end{align*}
The claimed multiplicative decrease is apparent. Moreover, it is seen that if $\tlam^{-1}_s > \Delta/5$, then 
$$
(\tlam^{-1}_s)_+ \ge \frac{1}{4} \frac{\Delta}{5} = \frac{\Delta}{20} > \frac{\Delta}{27}~.
$$
This completes our proof.
\end{proof}

\section{Principal Angles}
\begin{definition} \label{def:principalAngle}
Let $\cX$ and $\cY$ be two subspaces of $\reals^d$ of dimension at least $k$. The principal angles $0 \le \theta_1 \le \ldots \theta_k \le \frac{\pi}{2}$ between $\cX$ and $\cY$ and the corresponding principal pairs $(x_i,y_i)_{i=1}^k$ are defined by
\[
\theta_i  = \arccos(x_i^\top y_i) := \min \{\arccos (x^\top y): x \in \cX,~y \in \cY, \|x\|=1,\|y\|=1,~x \perp \{x_1,\ldots,x_{i-1}\},~x \perp \{y_1,\ldots,y_{i-1}\}\}~.
\]
The principal angles between matrices (whose columns are of the same size) are defined as the principal angles between their ranges.
\end{definition}
Following \cite{hardt2014noisy}, we use the following non-recursive expression.
\begin{lemma} \label{lem:nonRecursiveAngle}
Let $k \le p \le d$. Suppose that $S \in \reals^{d \times p}$ is a matrix of a full column rank and let $U = [U_k, U_{-k}] \in \cO^{d \times d}$. Then, 
\[
\tan(\theta_k(U_k,S)) = \min_{\Pi \in \cP_k} \max_{w: \Pi w=w } \frac{\|U_k^\top Sw\|}{\|U_{-k}^\top Sw\|}~, 
\]
\end{lemma}

\end{document}

%% file: header.tex
\usepackage{graphicx,subfigure}
\usepackage{epsfig}
\usepackage{hyperref}
\usepackage{amsmath}
\usepackage{amssymb}
\usepackage{algorithm}
\usepackage{algorithmic}
\usepackage{array}
\usepackage{eqparbox}
\usepackage{url}
\usepackage{enumerate}
\usepackage{amsfonts}
\usepackage{boxedminipage}
\usepackage{xcolor}
 \usepackage{framed}  
\usepackage{rotating}
\usepackage{array}
\usepackage{multirow}
\usepackage{color}
\usepackage{tikz}
\usepackage{pgfplots}
\usepackage{mathabx}
\usepackage{tabularx,ragged2e,booktabs,caption}

{
\begin{center}
\begin{boxedminipage}{0.8\linewidth}
\begin{center}
\textbf{\texttt{#1}}
\end{center}
\rm
\begin{tabbing}
....\=...\=...\=...\=...\=  \+ \kill
} %
{\end{tabbing} 
\end{boxedminipage} \end{center} 
}

{
\vspace{0.01cm}
\begin{center}
\begin{boxedminipage}{0.8\linewidth}
\begin{center}
\textbf{\texttt{#1}}
\end{center}
} %
{
\end{boxedminipage} \end{center} \vspace{0.3cm}
}

\newtheorem{definition}{Definition}
\newtheorem{lemma}{Lemma}
\newtheorem{corollary}{Corollary}
\newtheorem{theorem}{Theorem}

\newcommand{\reals}{\mathbb{R}}

\newcommand{\barLam}{\bar{\Lambda}}

\newcommand{\poly}{\mathrm{poly}}

\newcommand{\tu}{\tilde{u}}

\newcommand{\tU}{\tilde{U}}

\newcommand{\tO}{\tilde{O}}

\newcommand{\tG}{\tilde{G}}

\newcommand{\tlam}{\tilde{\lambda}}
\newcommand{\tPi}{\tilde{\Pi}}

\newcommand{\hU}{\hat{U}}

\newcommand{\hlam}{\hat{\lambda}}

\newcommand{\cY}{\mathcal{Y}}
\newcommand{\cX}{\mathcal{X}}

\newcommand{\cP}{\mathcal{P}}

\newcommand{\cN}{\mathcal{N}}

\newcommand{\cO}{\mathcal{O}}

\newenvironment{proof}{\par\noindent{\bf Proof\ }}{\hfill\BlackBox\\[2mm]}
\newcommand{\BlackBox}{\rule{1.5ex}{1.5ex}}

\makeatletter
\def\moverlay{\mathpalette\mov@rlay}
\def\mov@rlay#1#2{\leavevmode\vtop{%
   \baselineskip\z@skip \lineskiplimit-\maxdimen
   \ialign{\hfil$\m@th#1##$\hfil\cr#2\crcr}}}
\newcommand{\charfusion}[3][\mathord]{
    #1{\ifx#1\mathop\vphantom{#2}\fi
        \mathpalette\mov@rlay{#2\cr#3}
      }
    \ifx#1\mathop\expandafter\displaylimits\fi}
\makeatother

\newcommand{\hPi}{\hat{\Pi}}

\newcommand{\nnz}{\mathrm{nnz}}

\DeclareMathOperator*{\argmin}{arg\,min}

\renewcommand{\eqref}[1]{Equation~(\ref{#1})}

\newcommand{\secref}[1]{Section~\ref{#1}}

\newcommand{\thmref}[1]{Theorem~\ref{#1}}
\newcommand{\lemref}[1]{Lemma~\ref{#1}}

\newcommand{\corref}[1]{Corollary~\ref{#1}}

\newcommand{\algref}[1]{Algorithm~\ref{#1}}

\newcommand{\handout}[5]{
   \renewcommand{\thepage}{#1-\arabic{page}}
   \noindent
   \begin{center}
   \framebox{
      \vbox{
    \hbox to 5.78in { {\bf (67577) Introduction to Machine Learning}
         \hfill #2 }
       \vspace{4mm}
       \hbox to 5.78in { {\Large \hfill #5  \hfill} }
       \vspace{2mm}
       \hbox to 5.78in { {\it #3 \hfill #4} }
      }
   }
   \end{center}
   \vspace*{4mm}
}

